\documentclass[12pt]{article}
\usepackage{sbc-template}
\usepackage[utf8]{inputenc}
\usepackage[brazil]{babel}
\usepackage{amsmath}
\usepackage{array, amsfonts}
\usepackage{amsthm}
\usepackage{amssymb}
\usepackage{enumerate}
\usepackage{lineno}
\usepackage[bookmarks=true]{hyperref}
\usepackage{enumitem}
\usepackage{mathrsfs}
\usepackage{stackengine}
\usepackage{bookmark}
\usepackage{amsrefs}
\usepackage{color,soul}
\usepackage[draft]{todonotes}
\usepackage{MnSymbol}
\usepackage{algorithm2e}
\usepackage[noend]{algpseudocode}
\usepackage{mathtools}

\definecolor{darkgreen}{rgb}{0,0.5,0}

\newcommand{\comment}[1]

\newcommand{\RNum}[1]{\uppercase\expandafter{\romannumeral #1\relax}}

\title{ Transtemporal edges and crosslayer edges in incompressible high-order networks\footnote{ These results are also contained in \cite{Abrahao2018earxiv}, available at \url{https://arxiv.org/abs/1812.01170}. } }

\author{Felipe S. Abrah\~{a}o\inst{1}, Klaus Wehmuth\inst{1}, Artur Ziviani\inst{1}}

\address{National Laboratory for Scientific Computing (LNCC)
	\\ 25651-075 – Petropolis, RJ – Brazil
\email{\{fsa,klaus,ziviani\}@lncc.br}
}

\newtheorem{theorem}{Theorem}[section]

\newtheorem{corollary}{Corollary}[section]

\theoremstyle{definition}

\theoremstyle{remark}

\begin{document}

\maketitle

\begin{abstract}
	This work presents some outcomes of a theoretical investigation of incompressible high-order networks defined by a generalized graph representation. 
	We study some of their network topological properties and how these may be related to real-world complex networks. 
	We show that these networks have very short diameter, high k-connectivity, degrees of the order of half of the network size within a strong-asymptotically dominated standard deviation, and rigidity with respect to automorphisms. 
	In addition, we demonstrate that incompressible dynamic (or dynamic multilayered) networks have transtemporal (or crosslayer) edges and, thus,
	a snapshot-like representation of dynamic networks is inaccurate for capturing the presence of such edges that compose underlying structures of some real-world networks.
\end{abstract}

%


\section{Introduction}\label{sectionIntro}

The general scope of this paper is to study (plain) algorithmically random high-order networks. 
In a general sense, a high-order network is any network that has additional representational structures.
For example, this is the case of dynamic  (i.e., time-varying) networks, multilayer networks, and dynamic multilayer networks \cite{Wehmuth2018b,Wehmuth2016b}. 
Thus, as the interest and pervasiveness of complex network modelling and network analysis increase, the importance of accurate representations of such networks into new extensions of graph-theoretical abstractions has become of increasing importance.

Within the theoretical framework of algorithmic information theory, complex net\-works theory, and graph theory, we study incompressibility (i.e., algorithmic randomness) and computably irreducible information content (i.e., plain or prefix algorithmic complexity) in generalized graph representations. 
In particular, we are grounding our formalizations, methods, and results on \cite{Buhrman1999,Zenil2018a}. 
Such an approach to network complexity, lossless compression, and random graphs has identified useful tools to find, estimate, or measure underlying topological structures or properties, e.g., degree distribution, k-connectivity, diameter, and symmetries \cite{Zenil2018a}.
Moreover, it is also related to more traditional approaches from statistical (entropy-like) information theory \cite{Zenil2018a}. 
For example, from the classical noiseless coding theorem \cite{Zenil2018a}, we know that, for large enough $ n  = \left| \mathrm{V}(G) \right|$, every recursively labeled random graph $ G $ on $n$ vertices and edge probability $ p = 1/2 $ in the classical model $ \mathcal{G}(n,p) $ is expected to be incompressible (i.e., algorithmically random).  
However, only algorithmic information theory gives us tools for studying incompressibility of fixed individual graphs that are not generated or defined by stochastic processes \cite{Zenil2018a}.
In this way, this work derives from previous works on topological properties of incompressible graphs obtained from an algorithmic complexity (or algorithmic randomness) analysis of recursively labeled graphs.
Thus, these methods differ from the traditional methods in random graphs theory, such as the probabilistic method. 

In the present work, we apply the results on labeling and algorithmic randomness introduced in \cite{Abrahao2018d}, which extends those in \cite{Buhrman1999} to MultiAspect Graphs (MAGs) \cite{Wehmuth2016b}. 
MAGs are formal representations of dyadic (or $2$-place) relations between two arbitrary $n$-ary tuples and have shown fruitful representational properties to network modelling and analysis of high-order networks \cite{Wehmuth2016b,Wehmuth2018b}.
More formally, $ \mathscr{G}=(\mathscr{A},\mathscr{E}) $ denotes a MAG, where $\mathscr{E}$ is the set of existing composite edges (which are ordered $2p$-tuples) between two arbitrary composite vertices (which are ordered $p$-tuples) of $ \mathscr{G} $. 
Each aspect $ \mathbf{ \sigma } \in \mathscr{A} $ is a finite set,
and the number of aspects $ p = | \mathscr{A} | $ is called the \emph{order} of $ \mathscr{G} $. 
Thus, in the present work, we are assuming dynamic networks, multilayered networks, or dynamic multilayered networks as special cases of MAGs \cite{Wehmuth2016b}. 
To tackle these problems in the present paper, we apply a theoretical approach by presenting definitions and theorems. 


\section{From some topological properties of incompressible high-order networks to transtemporal and crosslayer edges}\label{sectionTopologicalTVG}


As defined in \cite{Wehmuth2016b}, a simple time-varying graph (TVG) is an undirected TVG without self-loops in the same way that a simple MAG (or graph) is an undirected MAG (or graph) without self-loops.
From \cite{Buhrman1999}, the randomness deficiency measures how the recursively labeled MAG (or graph)\footnote{ Once a graph is a first-order MAG. } diverges from being plain algorithmically random given $ \left| \mathbb{V}( \mathscr{G}_c ) \right|  $ such that 
\[ C\left( \mathscr{E}( \mathscr{G}_c ) \, \mid \left| \mathbb{V}( \mathscr{G}_c ) \right| \right) 
\geq 
\binom{ \left| \mathbb{V}( \mathscr{G}_c ) \right| }{ 2 } - \delta( \left| \mathbb{V}( \mathscr{G}_c ) \right| )  \]
holds and, in this case, we say $ \mathscr{G}_c  $ is $ \delta( \left| \mathbb{V}( \mathscr{G}_c ) \right| ) $-C-random.
Note that: 
$ C(  y | x) $ denotes the plain algorithmic complexity of $ y $ given $ x $;
$ \mathscr{G}_c   = \left( \mathscr{A} , \mathscr{E} \right)$ denotes an arbitrary simple MAG; 
$ \mathbb{V}( \mathscr{G}_c )  \coloneqq \bigtimes_{i=1}^{ \left| \mathscr{A} \right| } \mathscr{A}( \mathscr{G} )[i] $ is the set of all possible composite vertices of $ \mathscr{G}_c $; 
$ \mathscr{A} $ is a class (or list) of sets $ \mathscr{A}( \mathscr{G} )[i] $ such that each $i$, where $ 1 \leq i \leq \left| \mathscr{A} \right|  $, is an aspect;
and $ \delta( \left| \mathbb{V}( \mathscr{G}_c ) \right| )  $ is the randomness deficiency of $ \mathscr{G}_c $.
By convention, we can assume $ \mathrm{ V }( \mathscr{G} ) = \mathscr{A}( \mathscr{G} )[1]  $ as the set of vertices of the MAG, $ \mathrm{ T }( \mathscr{G} ) = \mathscr{A}( \mathscr{G} )[2] $ as the set of time instants, and further sets $ \mathscr{A}( \mathscr{G} )[i] , i  \geq 3 $ as a set of layers (of type $i$).

Since a TVG $ \mathrm{ G_t }=(\mathrm{V},\mathscr{E},\mathrm{T}) $ is a second order MAG \cite{Wehmuth2016b}, where $\mathrm{V}$ is the set of vertices, $\mathrm{T}$ is the set of time instants, and $\mathscr{E} \subseteq \mathrm{V} \times \mathrm{T} \times \mathrm{V} \times \mathrm{T}$ is the set of (composite) edges, it is immediate to show in Corollary~\ref{corC-randomTVGtopologicalproperties} that the previously studied case for simple MAGs with arbitrary randomness deficiency $ \delta( \left| \mathbb{V}( \mathscr{G}_c ) \right|  ) $  in \cite{Abrahao2018d} also applies to simple TVGs.
To this end, note that in \cite{Abrahao2018d}, since the order of the simple MAG is arbitrary, there is a recursively labeled infinite family of simple TVGs that satisfy this with a chosen randomness deficiency  $ \delta( \left| \mathbb{V}( \mathrm{ G_t } ) \right|  ) = \mathbf{O}( \log_2( \left| \mathbb{V}( \mathrm{ G_t } ) \right| ) ) $.
Thus, from Theorem 5.1 and Corollary 5.1.1 in \cite{Abrahao2018d} we have that: 

\begin{corollary}\label{corC-randomTVGtopologicalproperties}
	Let $ F_{ \mathrm{ G_t } }  $ be a recursively labeled infinite family $  F_{ \mathrm{ G_t } } \neq \emptyset  $ of simple TVGs $ \mathrm{ G_t } $ that are $ \mathbf{O}( \log_2( \left| \mathbb{V}( \mathrm{ G_t } ) \right| ) ) $-C-random.
	Then, the following hold for large enough $ \mathrm{ G_t } \in F_{ \mathrm{ G_t } } $, where $ \mathbb{V}( \mathrm{ G_t } ) = \mathrm{V}( \mathrm{ G_t } ) \times \mathrm{T}( \mathrm{ G_t } )  $:
	\begin{enumerate}
		\item The degree $ \mathbf{d}( \mathbf{v} ) $ of a composite vertex $ \mathbf{v} \in \mathbb{V}( \mathrm{ G_t } ) $ in a MAG $  \mathrm{ G_t } \in F_{ \mathrm{ G_t } } $ satisfies
		\[
		\left| \mathbf{d}( \mathbf{v} ) - \left( \frac{ \left| \mathbb{V}( \mathrm{ G_t } ) \right| - 1 }{ 2 } \right) \right| 
		= 
		\mathbf{O}\left( \sqrt{ \left| \mathbb{V}( \mathrm{ G_t } ) \right| \, \left(  \mathbf{O}( \log_2( \left| \mathbb{V}( \mathrm{ G_t } ) \right|) ) \right) } \right) \text{ .}
		\]
		\label{corK-randomMAGsproperties1}
		
		\item $  \mathrm{ G_t }  $ has $ \frac{\left| \mathbb{V}( \mathrm{ G_t } ) \right| }{4} + \mathbf{o}(\left| \mathbb{V}( \mathrm{ G_t } ) \right| ) $ disjoint paths of length 2 between each pair of composite vertices $ \mathbf{u} , \mathbf{v} \in \mathbb{V}( \mathrm{ G_t } ) $. 
		\label{corK-randomMAGsproperties2}
		
		\item $  \mathrm{ G_t } $ has (composite) diameter $2$.
		\label{corK-randomMAGsproperties3}

		\item $  \mathrm{ G_t } $ is rigid under permutations of composite vertices. 
		\label{corK-randomMAGsproperties5} 
	\end{enumerate}
	
\end{corollary}

In fact, this Corollary~\ref{corC-randomTVGtopologicalproperties} can be rewritten for arbitrary $ \mathbf{O}( \log_2( \left| \mathbb{V}( \mathscr{G}_c ) \right| ) ) $-C-random simple MAGs $ \mathscr{G}_c  $ instead of simple $ \mathbf{O}( \log_2( \left| \mathbb{V}( \mathrm{ G_t } ) \right| ) ) $-C-random TVGs $ \mathrm{ G_t } $ and, therefore, it also holds for other high-order networks, e.g.,
dynamic multilayer networks.

Now, let a \emph{transtemporal} edge be a composite edge $ e=( u , t_i , v , t_j ) \in \mathscr{E}( \mathrm{ G_t } ) $ with $ j \neq i \pm 1 $ and  $ j \neq i $. Thus, in Theorem~\ref{thmTranstemporaledges}, the short (composite) diameter and high $k$-connectivity (as defined in \cite{Buhrman1999}) of a $ \mathbf{O}( \log_2( \left| \mathbb{V}( \mathrm{ G_t } ) \right| ) ) $-C-random simple TVG ensures the existence of transtemporal edges in $ \mathrm{ G_t } $: 

\begin{theorem}\label{thmTranstemporaledges}
	Let $ \mathrm{ G_t } $ be a simple TVG satisfying Corollary~\ref{corC-randomTVGtopologicalproperties} with $ \left| \mathrm{ T }( \mathrm{ G_t }  ) \right| > 8  $.
	Then, for every pair of vertices $ u , v \in \mathrm{ V }( \mathrm{ G_t }  ) $ and time instants $ t_i , t_j \in \mathrm{ T }( \mathrm{ G_t }  ) $ with $ j > i+2 $, there is a transtemporal edge $ e \in \mathscr{ E }( \mathrm{ G_t } ) $. 
\end{theorem}
\begin{proof}[Proof]
	The case in which $ ( u , t_i , v , t_j ) \in \mathscr{ E }( \mathrm{ G_t } )  $ immediately satisfies the definition of transtemporal edge. From Corollary~\ref{corC-randomTVGtopologicalproperties}, the composite diameter is $2$. Therefore, it only remains to investigate the case in which there are $ h \in \mathrm{ V }( \mathrm{ G_t }  ) $ and $ t_z \in \mathrm{ T }( \mathrm{ G_t }  ) $ such that $ ( u , t_i , h , t_z ) \in \mathscr{ E }( \mathrm{ G_t } )  $ and $ ( h , t_z , v , t_j ) \in \mathscr{ E }( \mathrm{ G_t } )  $.
	From Corollary~\ref{corC-randomTVGtopologicalproperties}, we have that, for every pair of vertices $ u , v \in \mathrm{ V }( \mathrm{ G_t }  ) $ and time instants $ t_i , t_j \in \mathrm{ T }( \mathrm{ G_t }  ) $, there are $ \frac{\left| \mathbb{V}( \mathrm{ G_t } ) \right| }{4} + \mathbf{o}(\left| \mathbb{V}( \mathrm{ G_t } ) \right| ) $ disjoint paths of length 2 between $ ( u , t_i ) $ and $ ( v , t_j ) $.
	But, since $ \left| \mathrm{ T }( \mathrm{ G_t }  ) \right| > 8  $, the number of possible distinct composite vertices $ ( h , t_z ) $ with $ t_z = t_i $ or $ t_z = t_j $ will be always smaller than $ \frac{\left| \mathbb{V}( \mathrm{ G_t } ) \right| }{4} $ and, thus, strictly smaller than the number of distinct composite vertices connecting $ ( u , t_i ) $ and $ ( v , t_j ) $. 
	Therefore, there will be at least one composite vertex $ ( h , t_z ) $ with $ i + 1 < z $, $ \, z + 1 < j $, $ \, z < i $, or $ j < z $. Then, in any case, $ ( u , t_i , h , t_z ) \in \mathscr{ E }( \mathrm{ G_t } )  $ or $ ( h , t_z , v , t_j ) \in \mathscr{ E }( \mathrm{ G_t } )  $ will be a transtemporal edge.
\end{proof}

In fact, Theorem~\ref{thmTranstemporaledges} can easily be generalized to multilayer (undirected) networks or dynamic multilayered (undirected) networks, so that Theorem~\ref{thmTranstemporaledges} will become a corollary.
To this end, it suffices to extend Corollary~\ref{corC-randomTVGtopologicalproperties} to simple MAGs with order $ p \geq 2 $. 
First, the multilayer case in which there is just one additional aspect, besides the set of vertices, is totally analogous. 
Secondly, for the dynamic multilayer (or many-type multilayer) case in which the simple MAGs have order $ p > 2 $, we will have that the first aspect still is the set of vertices, the second aspect still is the set $ \mathrm{ T }( \mathscr{ G }_c ) =  \mathscr{A}( \mathscr{ G }_c  )[2]  $ of time instants (or the first layer type $ \mathrm{ L_2 }( \mathscr{ G }_c ) = \mathscr{A}( \mathscr{ G }_c  )[2] $), and the further aspects are any other layer type $ \mathrm{ L_k }( \mathscr{ G }_c ) = \mathscr{A}( \mathscr{ G }_c  )[k] $, where $ k > 2 $.  
Analogously to the temporal case, let a \emph{crosslayer} edge be a composite edge $ e=( u ,  \dots  , x_{hi} , \dots ,  x_{ps} , v , \dots , x_{hj} \dots , x_{ps'}) \in \mathscr{E}( \mathscr{ G }_c  ) $ in which $ j \neq i \pm 1 $ and  $ j \neq i $, where $ 2 \leq h \leq p $ and $ x_{hi} , x_{hj} \in \mathscr{A}( \mathscr{ G }_c  )[h] $. Thus, Theorem~\ref{thmTranstemporaledges} can be rewritten as: 

\begin{theorem}\label{lemmaTransaspectedges}
	Let $ \mathscr{ G }_c $ be a $ \mathbf{O}( \log_2( \left| \mathbb{V}( \mathscr{ G }_c ) \right| ) ) $-C-random simple MAG with order $ p \geq 2 $ that belongs to a recursively labeled infinite family $  F_{ \mathscr{ G }_c } $ of simple MAGs $ \mathscr{ G }_c $ such that 
	\[
	\left| \mathscr{A}( \mathscr{ G }_c  )[k]  \right|
	=
	\frac{\left| \mathbb{ V }( \mathscr{ G }_c  ) \right| }{ \left| \mathrm{ V }( \mathscr{ G }_c ) \right| \, 
		\bigtimes\limits_{ h \geq 2 , \, h \neq k } \left| \mathscr{A}( \mathscr{ G }_c  )[h]  \right|
	} 
	> 
	8
	\]
	Then, for every pair of composite vertices $ ( u ,  \dots  , x_{ki} , \dots ,  x_{ps}  ) $ and $ ( v , \dots , x_{kj} \dots , x_{ps'} ) $ with $ j > i+2 $, where $ 2 \leq k \leq p $ and $ x_{ki} , x_{kj}  \in \mathscr{A}( \mathscr{ G }_c  )[k] $, there is a crosslayer edge $ e \in \mathscr{ E }( \mathscr{ G }_c ) $, if $ \mathrm{ L_k }( \mathscr{ G }_c )  = \mathscr{A}( \mathscr{ G }_c  )[k] $, or a transtemporal edge $ e \in \mathscr{ E }( \mathscr{ G }_c ) $, if $ \mathrm{ T }( \mathscr{ G }_c )   = \mathscr{A}( \mathscr{ G }_c  )[k] $. 
	
\end{theorem}

\section{Conclusions}\label{sectionConclusions}

In this work, we have studied some topological properties of plain algorithmically random high-order networks that can be formally represented by MultiAspect Graphs (MAGs), 
in particular, dynamic networks, multilayer networks, and dynamic multilayer networks.
We have shown that these networks have very short diameter, high k-connectivity, degrees of the order of half of the network size within a strong-asymptotically dominated standard deviation, and rigidity under permutations of composite vertices. 
Therefore, these theoretical findings directly relate lossless compressibility of high-order networks with their network topological properties. 

Then, we have demonstrated the presence of transtemporal or crosslayer edges (i.e., edges linking vertices at non-adjacent time instants or layers) in incompressible dynamic, multilayer, or dynamic multilayer networks. 
Note that a snapshot-like representation of dynamic networks is not accurate enough to capture the presence of such edges in the underlying structures of some real-world networks.
Thus, with the purpose of bringing algorithmic randomness to the context of high-order networks or complex networks, our theoretical results suggest that estimating or analyzing both the incompressibility and the network topological properties of real-world networks cannot be taken into a universal approach, such as the incompressibility of arbitrary MAGs. 

\section*{Acknowledgements}

The authors are grateful for the support of CAPES, CNPq, FAPERJ, and FAPESP.

\bibliographystyle{plain}
\bibliography{ETC2019-2.1-CompleteRefs-Felipe.bib}

\begin{bibdiv}
\begin{biblist}

\bib{Abrahao2018d}{techreport}{
      author={Abrah{\~{a}}o, Felipe~S.},
      author={Wehmuth, Klaus},
      author={Ziviani, Artur},
       title={{An algorithmically random family of MultiAspect Graphs and its
  topological properties}},
        type={Research Report},
    language={English},
   publisher={National Laboratory for Scientific Computing (LNCC)},
 institution={National Laboratory for Scientific Computing (LNCC)},
     address={Petr{\'{o}}polis},
        date={2018},
  url={http://www.lncc.br/departamentos/producaocientificageral.php?vMenu=2{\&}vTipo=13{\&}vCabecalho=pesq{\&}vTitulo=lncc{\&}vDepto={\&}idt{\_}responsavel={\&}vAn
  http://arxiv.org/abs/1810.11719},
        note={Research Report no. 8/2018. Available at:
  \url{https://doi.org/10.5281/zenodo.1472893}},
}

\bib{Abrahao2018earxiv}{unpublished}{
      author={Abrah{\~{a}}o, Felipe~S.},
      author={Wehmuth, Klaus},
      author={Ziviani, Artur},
       title={{On incompressible high order networks}},
   publisher={arXiv Preprints},
        date={2018},
         url={http://arxiv.org/abs/1812.01170},
        note={Available at \url{https://arxiv.org/abs/1812.01170}},
}

\bib{Buhrman1999}{article}{
      author={Buhrman, Harry},
      author={Li, Ming},
      author={Tromp, John},
      author={Vit{\'{a}}nyi, Paul},
       title={{Kolmogorov Random Graphs and the Incompressibility Method}},
        date={1999jan},
     journal={SIAM Journal on Computing},
      volume={29},
      number={2},
       pages={590\ndash 599},
}

\bib{Wehmuth2016b}{article}{
      author={Wehmuth, Klaus},
      author={Fleury, {\'{E}}ric},
      author={Ziviani, Artur},
       title={{On MultiAspect graphs}},
        date={2016},
        ISSN={03043975},
     journal={Theoretical Computer Science},
      volume={651},
       pages={50\ndash 61},
}

\bib{Wehmuth2018b}{inproceedings}{
      author={Wehmuth, Klaus},
      author={Ziviani, Artur},
       title={{Centralities in High Order Networks}},
    language={English},
   booktitle={{Meeting on Theory of Computation (ETC), Congress of the
  Brazilian Computer Society (CSBC) 2018}},
}

\bib{Zenil2018a}{article}{
      author={Zenil, Hector},
      author={Kiani, Narsis},
      author={Tegn{\'{e}}r, Jesper},
       title={{A Review of Graph and Network Complexity from an Algorithmic
  Information Perspective}},
        date={2018jul},
        ISSN={1099-4300},
     journal={Entropy},
      volume={20},
      number={8},
       pages={551},
}

\end{biblist}
\end{bibdiv}

\end{document}